\documentclass[11pt]{article}

\usepackage{graphicx, amssymb, latexsym, amsfonts, amsmath, lscape, amscd,
amsthm, color, epsfig, mathrsfs, tikz, enumerate}

\newtheorem{theorem}{Theorem}[section]

\newtheorem{lemma}[theorem]{Lemma}

\newtheorem{observation}{Observation}

\newcommand\DELETE[1]{}

\begin{document}

\title{{\bf On chromatic number of colored mixed graphs}}
\author{
{\sc Sandip Das\footnote{{\small sandipdas@isical.ac.in}}}, {\sc Soumen Nandi\footnote{{\small soumen.nandi\_r@isical.ac.in}}}, {\sc Sagnik Sen\footnote{{\small sen007isi@gmail.com}}}\\
\mbox{}\\
{\small Indian Statistical Institute, Kolkata, India}\\
}

%{\sc Sandip Das\footnote{\small sandipdas@isical.ac.in}}}, {\sc Soumen Nandi\footnote{\small soumen.nandi\_r@isical.ac.in}}}, {\sc Sagnik Sen\footnote{\small sen007isi@gmail.com}}}\\
%%

\date{\today}

\maketitle

% ----------------------------------------------------------------
\begin{abstract}
An $(m,n)$-colored mixed graph $G$ is a graph  with its arcs having one of the $m$ different colors and edges having one of the $n$ different colors. A homomorphism $f$ of an $(m,n)$-colored mixed graph $G$   to an  $(m,n)$-colored mixed graph $H$ is a vertex mapping such that if $uv$ is an arc (edge) of color $c$ in $G$, then $f(u)f(v)$ is an arc (edge) of  color $c$ in $H$. The \textit{$(m,n)$-colored mixed chromatic number} $\chi_{(m,n)}(G)$ of an   $(m,n)$-colored mixed graph $G$ is the 
order (number of vertices)
of the smallest homomorphic image of $G$. This notion was introduced by 
Ne\v{s}et\v{r}il and Raspaud (2000,  J. Combin. Theory, Ser. B 80, 147--155). 
They showed that $\chi_{(m,n)}(G) \leq k(2m+n)^{k-1}$ where $G$ is a $k$-acyclic colorable graph.  We proved the tightness of this bound. 
We also showed that the acyclic chromatic number of a graph is bounded by $k^2 + k^{2 + \lceil log_{(2m+n)} log_{(2m+n)} k \rceil}$ if its 
$(m,n)$-colored mixed chromatic number is at most $k$.  
Furthermore, using probabilistic method, we showed that for graphs with maximum degree $\Delta$ its $(m,n)$-colored mixed chromatic number is at most 
$2(\Delta-1)^{2m+n} (2m+n)^{\Delta-1}$. 
In particular, the last result directly improves the upper bound  $2\Delta^2 2^{\Delta}$ of oriented chromatic number of 
graphs with  maximum degree $\Delta$, obtained by Kostochka, Sopena and Zhu (1997, J. Graph Theory 24, 331--340)  to $2(\Delta-1)^2 2^{\Delta -1}$. 
We also show that there exists a graph with maximum degree $\Delta$ and $(m,n)$-colored mixed chromatic number at least $(2m+n)^{\Delta / 2}$. 
 \end{abstract}


\section{Introduction}
An \textit{$(m,n)$-colored mixed graph} $G = (V, A \cup E)$ is a graph $G$ with set of vertices $V$, set of arcs $A$ and set of edges $E$ where each arc is 
 colored by one of the $m$ colors $\alpha_1, \alpha_2, ..., \alpha_m$ and each edge is  colored by one of the $n$ colors $\beta_1, \beta_2, ..., \beta_n$. We denote  the number of vertices and the number of edges of the underlying graph of $G$ by $v_{G}$  and $e_{G}$, respectively.
Also, we will consider only those 
$(m,n)$-colored mixed graphs for which the underlying undirected graph is simple.
Ne\v{s}et\v{r}il and Raspaud~\cite{raspaud_and_nesetril} generalized the notion of vertex coloring and chromatic number 
 for $(m,n)$-colored mixed graphs
 by definining colored homomorphism.

Let $G = (V_1, A_1 \cup E_1)$ and $H = (V_2, A_2 \cup E_2)$
be two $(m,n)$-colored mixed graphs. 
A colored homomorphism of $G$ to $H$ is a function $f : V_1 \rightarrow V_2$ satisfying
$$uv \in A_1 \Rightarrow f(u)f(v) \in A_2,$$ 
$$uv \in E_1 \Rightarrow f(u)f(v) \in E_2,$$
and the color of the arc or  edge linking $f(u)$ and $f(v)$ is the same as the color of the arc or the edge linking $u$ and $v$~\cite{raspaud_and_nesetril}.
  We write $ G \rightarrow  H$ whenever there exists a 
 homomorphism of $ G$ to $ H$.

Given an $(m,n)$-colored mixed graph $G$ let $H$ be an $(m,n)$-colored mixed graph with minimum \textit{order} (number of vertices) such that $G \rightarrow H$. 
Then the order of $H$ is the \textit{$(m,n)$-colored mixed chromatic number} $\chi_{(m,n)}(G)$ of $G$.
For an undirected simple graph $G$, the maximum $(m,n)$-colored mixed chromatic number taken over all $(m, n)$-colored mixed
graphs having underlying undirected simple graph $G$ is denoted by  $\chi_{(m,n)}(G)$.
Let $\mathcal{F}$ be a family of undirected simple graphs. 
Then  $\chi_{(m,n)}(\mathcal{F})$ is the maximum of $\chi_{(m,n)}(G)$ 
taken over all $G \in \mathcal{F}$.

Note that a $(0,1)$-colored mixed graph $G$  is nothing but an undirected simple graph while  
$\chi_{(0,1)}(G)$ is the ordinary chromatic number. 
Similarly, the study of $\chi_{(1,0)}(G)$ 
is  the study of oriented chromatic number which is considered by several researchers in the last two decades (for details please check the recent updated survey~\cite{sopena_updated_survey}).
Alon and Marshall~\cite{Marshall-edgecoloring}  studied the homomorphism of $(0,n)$-colored mixed graphs with a particular focus on $n=2$. 

A simple graph $G$ is \textit{$k$-acyclic colorable} if we can color its vertices with $k$ colors such that each color class induces an independent set and any two color class induces a forest. 
The \textit{acyclic chromatic number} $\chi_a(G)$ of a simple graph $G$ is the minimum  $k$ 
such that $G$ is $k$-acyclic colorable. 
Ne\v{s}et\v{r}il and Raspaud~\cite{raspaud_and_nesetril} showed that $\chi_{(m,n)}(G) \leq k(2m+n)^{k-1}$ where $G$ is a $k$-acyclic colorable graph. 
As planar graphs are $5$-acyclic colorable due to Borodin~\cite{Borodinacyclic}, the same authors implied 
 $\chi_{(m,n)}(\mathcal{P}) \leq 5(2m+n)^4$ for the family $\mathcal{P}$ of planar graphs as a corollary.
This result, in particular, implies $\chi_{(1,0)}(\mathcal{P}) \leq 80$ and $\chi_{(0,2)}(\mathcal{P}) \leq 80$ (independently proved 
before in~\cite{planar80} and~\cite{Marshall-edgecoloring}, respectively).

 Let $\mathcal{A}_k$ be the family of graphs with acyclic chromatic number at most $k$. Ochem~\cite{Ochem_negativeresults} showed that the upper bound  $\chi_{(1,0)}(\mathcal{A}_k) \leq 80$   is tight.  
We generalize it for all $(m,n) \neq (0,1)$ to show that the upper bound $\chi_{(m,n)}(\mathcal{A}_k) \leq k(2m+n)^{k-1}$
 obtained by Ne\v{s}et\v{r}il and Raspaud~\cite{raspaud_and_nesetril}
is tight. 
This implies that the upper bound $\chi_{(m,n)}(\mathcal{P}) \leq 5(2m+n)^4$ cannot be improved using
the upper bound of $\chi_{(m,n)}(\mathcal{A}_5)$.

The arboricity $arb(G)$ of a graph $G$ is the minimum $k$ such that the edges of $G$ can be decomposed into $k$ forests. 
Kostochka, Sopena and Zhu~\cite{Kostochka97acyclicand} showed that given a simple graph $G$, the acyclic chromatic number $\chi_a(G)$ of $G$ is also bounded by a function of  $\chi_{(1,0)}(G)$. 
We generalize this result for all $(m,n) \neq (0,1)$ by showing that for a graph $G$ with 
$\chi_{(m,n)}(G) \leq k$ we have 
$\chi_a(G) \leq k^2 + k^{2 + \lceil log_2 log_p k \rceil}$ where $p = 2m+n$. Our bound slightly improves the bound obtained by  Kostochka, Sopena and Zhu~\cite{Kostochka97acyclicand} for $(m,n) = (1,0)$.
For achieving this result we first establish some relations among arboricity of a graph, $(m,n)$-colored mixed chromatic number and 
acyclic chromatic number.

 Let $\mathcal{G}_{\Delta}$ be the family of graphs with maximum degree  $\Delta$.
Kostochka, Sopena and Zhu~\cite{Kostochka97acyclicand} proved that $2^{\Delta/2}  \chi_{(1,0)}(\mathcal{G}_{\Delta}) \leq 2\Delta^2 2^{\Delta}$.
We improve this result in a generalized setting by proving $p^{\Delta/2}   \leq  \chi_{(m,n)}(\mathcal{G}_{\Delta}) \leq 2(\Delta-1)^p p^{\Delta-1}$ for all  $(m,n) \neq (0,1)$ where 
$p = 2m+n$.

\section{Preliminaries}

A \textit{special 2-path} $uvw$ of an $(m,n)$-colored mixed graph $G$ is a 2-path satisfying one of the following conditions:

\begin{itemize}
\item[(i)] $uv$ and $vw$ are edges of different colors,

\item[(ii)] $uv$ and $vw$ are arcs (possibly of the same color),

\item[(iii)] $uv$ and $wv$ are arcs of different colors,

\item[(iv)] $vu$ and $vw$ are arcs of different colors,

\item[(v)] exactly one of $uv$ and $vw$ is an edge and the other is an arc.
\end{itemize}

\begin{observation}\label{special}
The endpoints of a special 2-path must have different image under any homomorphism of $G$.
\end{observation}

\begin{proof}
Let $uvw$ be a special 2-path in an $(m,n)$-colored mixed graph $G$. Let $f: G \rightarrow H$ be a colored homomorphism of $G$ to an
$(m,n)$-colored mixed graph $H$ such that $f(u) = f(w)$. Then $f(u)f(v)$ and $f(w)f(v)$  will induce parallel edges in the underlying graph of $H$.
But as we are dealing with $(m,n)$-colored mixed graphs with underlying simple graphs, this is not possible. 
\end{proof}

Let $G = (V, A \cup E)$ be an $(m,n)$-colored mixed graph. Let $uv$ be an arc of $G$ with color $\alpha_i$ for some $i \in \{1,2,...,m\}$. 
Then $u$ is a
\textit{$-\alpha_i$-neighbor} of $v$ and $v$ is a \textit{$+\alpha_i$-neighbor} of $u$. The set of all $+\alpha_i$-neighbors 
and $-\alpha_i$-neighbors of $v$ is denoted by $N^{+\alpha_i}(v)$ and $N^{-\alpha_i}(v)$, respectively. Similarly, 
let $uv$ be an edge of $G$ with color $\beta_i$ for some $i \in \{1,2,...,n\}$. Then $u$ is a
\textit{$\beta_i$-neighbor} of $v$ and the set of all $\beta_i$-neighbors of $v$ is denoted by $N^{\beta_i}(v)$.
Let $\vec{a} = (a_1, a_2, ..., a_j)$ be a \textit{$j$-vector} such that 
$a_i \in \{\pm \alpha_1, \pm \alpha_2, ...,  \pm \alpha_m,  \beta_1, \beta_2, ..., \beta_n\}$ where $i \in \{1,2,...,j\}$. 
Let $J = (v_1, v_2, ..., v_j)$ be a \textit{$j$-tuple} (without repetition) of  vertices  from $G$. Then we define the set 
$N^{\vec{a}}(J) = \{v \in V | v \in N^{a_i}(v_i) \text{ \textit{for all} } 1 \leq i \leq j\} $. 
Finally, we say that $G$ has property $Q^{t,j}_{g(j)}$ if for each $j$-vector $\vec{a}$ and each $j$-tuple $J$ we have 
$|N^{\vec{a}}(J)| \geq g(j)$ where $j \in \{0,1,...,t\}$ and $g: \{0,1,...,t\} \rightarrow \{0,1,... \infty\}$ is an integral function.

  \section{On graphs with bounded acyclic chromatic number}\label{sec acyclic}
  First we will construct examples of $(m,n)$-colored mixed graphs $H_k^{(m,n)}$ with acyclic chromatic number at most $k$ and 
  $\chi_{(m,n)}(H_k^{(m,n)}) = k(2m+n)^{k-1}$ for all $k \geq 3$ and for all $(m,n) \neq (0,1)$. This, along with the upper bound established by Ne\v{s}et\v{r}il and Raspaud~\cite{raspaud_and_nesetril}, will imply the following result:

  \begin{theorem}\label{acyclic-chromatic}
  Let $\mathcal{A}_k$ be the family of graphs with acyclic chromatic number at most $k$. 
  Then $\chi_{(m,n)}(\mathcal{A}_k) = k(2m+n)^{k-1}$ for all $k \geq 3$ and for all $(m,n) \neq (0,1)$. 
  \end{theorem}
  
\begin{proof}
First we will construct an $(m,n)$-colored mixed graph $H_k^{(m,n)}$, where $p = 2m+n \geq 2$, as follows. 
Let $A_{k-1}$ be the set of all $(k-1)$-vectors. 
Thus, $|A_{k-1}| = p^{k-1}$.  

Define $B_i$ as a set of $(k-1)$ vertices $B_i = \{b^i_1, b^i_2, ...,b^i_{k-1}\}$ for all $i \in \{1, 2, ..., k\}$ such that $B_r \cap B_s = \emptyset$ when $r \neq s$. The vertices of $B_i$'s are called \textit{bottom} vertices for each $i \in \{1, 2, ..., k\}$.
Furthermore, let $TB_i = (b^i_1, b^i_2, ...,b^i_{k-1})$ be a $(k-1)$-tuple.  

After that define the set of  vertices $T_i = \{t^i_{\vec{a}}| t^i_{\vec{a}} \in N^{\vec{a}}(TB_i) \text{ \textit{for all} } \vec{a} \in A_{k-1}\}$ 
for all $i \in \{1, 2, ..., k\}$.  
The vertices of $T_i$'s are called \textit{top} vertices for each $i \in \{1, 2, ..., k\}$. 
Observe that there are $p^{k-1}$ vertices in $T_i$ for each $i  \in \{1, 2, ..., k\}$.

Note that the definition of $T_i$ already implies some colored arcs and edges between the set of vertices 
$B_i$ and $T_i$ for all   $i \in \{1, 2, ..., k\}$.

As $p \geq 2$ it is possible to construct a special 2-path. 
Now for each  pair of vertices $u \in T_i$ and $v \in T_j$ ($i \neq j$), construct a special 2-path $uw_{uv}v$ and  call these new vertices $w_{uv}$ as 
\textit{internal} vertices for all $i, j \in \{1, 2, ..., k\}$.
This so obtained graph is $H_k^{(m,n)}$.

Now we will show that $\chi_{(m,n)}(H_k^{(m,n)}) \geq k(2m+n)^{k-1}$. 
Let $\vec{a} \neq \vec{a'}$ be two distinct 
$(k-1)$-vectors. 
Assume that the $j^th$ co-ordinate of $\vec{a}$ and  $\vec{a'}$ is different. 
Then note that $t^i_{\vec{a}}b^i_jt^i_{\vec{a'}}$ is a special 2-path. 
Therefore, $t^i_{\vec{a}}$ and $t^i_{\vec{a}}$ must have different homomorphic image under any homomorphism. 
Thus, all the vertices in $T_i$ must have distinct homomorphic image under any homomorphism. 
Moreover, as a vertex of $T_i$ is connected by a special 2-path with a vertex of $T_j$ for all $i \neq j$, all the top vertices must have distinct 
homomorphic image under any homomorphism. It is easy to see that $|T_i| =  p^{k-1}$ for all   $i \in \{1, 2, ..., k\}$.
Hence $\chi_{(m,n)}(H_k^{(m,n)}) \geq \sum_{i=1}^k |T_i| = k(2m+n)^{k-1}$.

 Then we will show that $\chi_{a}(H_k^{(m,n)}) \leq k$. From now on, by $H_k^{(m,n)}$, we mean the underlying undirected simple graph of the 
 $(m,n)$-colored mixed graph $H_k^{(m,n)}$. We will provide an acyclic coloring of this graph with $\{1,2,...,k\}$.
 Color all the vertices of $T_i$ with $i$ for all $i \in \{1,2,...,k\}$. Then color all the vertices of $B_i$ with distinct $(k-1)$ colors 
 from the set $\{1,2,...,k\} \setminus \{i\}$ of colors for all $i \in \{1,2,...,k\}$.
Note that each internal vertex have exactly two neighbors. Color each internal vertex with a color different from its neighbors.
It is easy to check that this is an acyclic coloring. 

Therefore, we showed that $\chi_{(m,n)}(\mathcal{A}_k) \geq k(2m+n)^{k-1}$ 
while, on the other hand, Ne\v{s}et\v{r}il and Raspaud~\cite{raspaud_and_nesetril} showed that $\chi_{(m,n)}(\mathcal{A}_k) \leq k(2m+n)^{k-1}$ for all $k \geq 3$ and for all $(m,n) \neq (0,1)$. 
  \end{proof}

Consider a complete graph  $K_t$. Replace all its edges  by a 2-path to obtain the graph $S$. 
For all $(m,n) \neq (0,1)$, it is possible to assign colored edges/arcs to 
the edges of $S$ such that it becomes an $(m,n)$-colored mixed graph with $t$ vertices that are pairwise connected by 
a special 2-path. Therefore, by Observation~\ref{special} we know that $\chi_{(m,n)}(S) \geq t$ whereas, it is easy to note that $S$ has arboricity 2. 
Thus, the $(m,n)$-colored mixed chromatic number  is not bounded by any function of arboricity.
Though the reverse type of bound exists. Kostochka, Sopena and Zhu~\cite{Kostochka97acyclicand}  proved such a bound for $(m,n) =(1,0)$. We generalize their result for all $(m,n) \neq (0,1)$.

\begin{theorem}\label{chromatic-arboricity}
Let $G$ be an $(m,n)$-colored mixed graph with $\chi_{(m,n)}(G) = k$ where $p = 2m+n \geq 2$. Then $arb(G) \leq \lceil log_{p}k +k/2 \rceil $. 
\end{theorem}

\begin{proof}
Let $G'$ be an arbitrary labeled subgraph of $G$ consisting $v_{G'}$ vertices and $e_{G'}$ edges. We know from 
Nash-Williams' Theorem~\cite{nash1page}  that the arboricity $arb(G)$ of any graph $G$ is equal to the maximum of 
$\lceil e_{G'}/(v_{G'} - 1) \rceil$ over all subgraphs $G'$ of $G$. So it is sufficient to prove that for any subgraph $G'$ of $G$, 
$e_{G'}/(v_{G'} - 1) \leq log_p k + k/2$. 
As $G'$ is a labeled graph, so there are $p^{e_{G'}}$ different $(m,n)$-colored mixed graphs with underlying graph $G'$. As $\chi_{(m,n)}(G) = k$, there exits a homomorphism from $G'$ to a $(m, n)$-colored mixed graph $G_k$ which has the complete graph on $k$ vertices as its underlying graph. 
Note that the number of possible homomorphisms of $G'$ to $G_k$ is at most $k^{v_{G'}}$.
For each such homomorphism of  $G'$ to $G_k$ there are at most $p^{k \choose 2}$ different $(m,n)$-colored mixed graphs with underlying labeled graph  $G'$ as there are $p^{k \choose 2}$ choices of $G_k$. 
Therefore,
\begin{equation}\label{eqn wolog}
p^{k \choose 2}. k^{v_{G'}} \geq p^{e_{G'}}
\end{equation}
  which implies 
  \begin{equation}\label{eqn wlog}
   log_p k \geq (e_{G'}/v_{G'}) - {k \choose 2} / v_{G'}.
\end{equation}  

If $v_{G'} \leq k$, then $e_{G'}/(v_{G'} - 1) \leq v_{G'}/2 \leq k/2$. Now let $v_{G'} > k$. We know that $\chi_{(m,n)}(G') \leq \chi_{(m,n)}(G) = k$. So

\begin{equation*} \label{.}
\begin{split}
log_p k & \geq \frac{e_{G'}}{v_{G'}} - \frac{k(k - 1)}{2 v_{G'}} \\
 & \geq \frac{e_{G'}}{(v_{G'} -1)} - \frac{e_{G'}}{v_{G'}(v_{G'} - 1)} - \frac{k - 1}{2} \\
 & \geq \frac{e_{G'}}{(v_{G'} -1)} - 1/2 - k/2 + 1/2 \\
 & \geq \frac{e_{G'}}{(v_{G'} -1)} - k/2.
\end{split}
\end{equation*}

Therefore, $\frac{e_{G'}}{(v_{G'} -1)} \leq log_{p}k +k/2$. 
\end{proof}

We have seen that the $(m,n)$-colored mixed chromatic number of a graph $G$ is bounded by a function of the acyclic chromatic number of $G$. Here we show that
 it is possible to bound the acyclic chromatic number of a graph in terms of its $(m,n)$-colored mixed chromatic number and arboricity.
Our result is a generalization of a similar result proved for $(m,n) =(1,0)$ by  Kostochka, Sopena and Zhu~\cite{Kostochka97acyclicand}.

\begin{theorem}\label{arboricity.chromatic-acyclic}
Let $G$ be an $(m,n)$-colored mixed graph with $arb(G) = r$ and $\chi_{(m,n)}(G) = k$ where $p = 2m+n \geq 2$.  Then 
$\chi_a(G) \leq k^{\lceil log_pr \rceil +1}$.
\end{theorem}

\begin{proof}
First we rename the following symbols:  
$\alpha_{1} = a_0, -\alpha_{1} = a_{1}, \alpha_2 = a_2, -\alpha_2 = a_3, ..., \alpha_m = a_{2m-2}, -\alpha_m = a_{2m-1}, \beta_1 = a_{2m}, \beta_2 = a_{2m+1},..., \beta_n = a_{2m+n-1}$.

Let $G$ be a graph with $\chi_{(m,n)}(G) = k$ where $2m+n = p$.
Let $v_1, v_2, ..., v_t$ be some ordering of the vertices of $G$. 
Now consider the $(m,n)$-colored mixed graph $G_0$ with underlying graph $G$ such that for any $i < j$ we have 
$v_j \in N^{a_0}(v_i)$ 
whenever $v_iv_j$ is 
an edge of $G$. 

Note that the edges of  $G$ can be covered by $r$ edge disjoint forests $F_1, F_2, ..., F_r$ as $arb(G) = r$. 
Let $s_i$ be the number $i$ expressed with base $p$
for all $i \in \{1,2,...,r\}$. Note that $s_i$ can have at most $s = \lceil log_pr \rceil$ digits.

  Now we will construct a sequence of $(m,n)$-colored mixed graphs $G_1, G_2, ..., G_s$ each having underlying graph $G$. 
  For a fixed $l \in \{1,2,...,s\}$ we will describe the construction of $G_l$.   
  Let $i <j$ and $v_iv_j$ is an edge of $G$. 
  Suppose $v_iv_j$ is an edge of the forest $F_{l'}$ for some $l' \in \{1,2,...,r\}$.
  Let the $l^{th}$ digit of $s_{l'}$  be $s_{l'}(l)$. Then  $G_l$ is constructed in a way such that 
  we have $v_j \in N^{a_{s_{l'}(l)}}(v_i)$ in $G_l$.
  
  Note that there is a homomorphism  $f_l: G_l \rightarrow H_l$ for each $l \in \{1,2,...,s\}$ such that $H_l$ is 
  an $(m,n)$-colored mixed graph on $k$ vertices. 
  Now we claim that $f(v) = (f_0(v), f_1(v), ..., f_s(v))$ for each $v \in V(G)$ is an acyclic coloring of $G$. 
  
  For adjacent vertices $u,v$ in $G$ clearly we have $f(v) \neq f(u)$ as $f_0(v) \neq f_0(u)$. 
  Let $C$ be a cycle in $G$. We have to show that at least 3 colors have been used to color this cycle with respect to the coloring given by $f$. 
  Note that in $C$ there must be two incident edges $uv$ and $vw$ such that they belong to different forests, 
  say, $F_i$ and $F_{i'}$, respectively.
 Now suppose that $C$ received two colors with respect to $f$. Then we must have $f(u) = f(w) \neq f(v)$. In particular we must have 
 $f_0(u) = f_0(w) \neq f_0(v)$. 
 To have that we must also have $u,w \in N^{a_i}(v)$ for some $i \in \{0,1,...,p-1\}$ in $G_0$. 
 Let  $s_i$ and $s_{i'}$ differ in their $j^{th}$ digit. Then in $G_j$ we have $u \in N^{a_i'}(v)$ and $w \in N^{a_i''}(v)$
 for some $i' \neq i''$. Then we must have $f_j(u) \neq f_j(w)$.  Therefore, we also have $f(u) \neq f(w)$. Thus, the cycle $C$ cannot be colored with two colors under the coloring $f$. So $f$ is indeed an acyclic coloring of $G$. 
    \end{proof}

Thus, combining   Theorem~\ref{chromatic-arboricity} and~\ref{arboricity.chromatic-acyclic} we have 
$\chi_a(G) \leq k^{\lceil log_p \lceil log_{p}k +k/2 \rceil \rceil +1}$ for $\chi_{(m,n)}(G) = k$ where $p = 2m+n \geq 2$. 
However,  we managed to obtain the following better bound.

\begin{theorem}\label{chromatic-acyclic}
Let $G$ be an $(m,n)$-colored mixed graph with  $\chi_{(m,n)}(G) = k \geq 4$ where $p = 2m+n \geq 2$.  Then 
$\chi_a(G) \leq k^2 +k^{2+ \lceil log_2log_pk \rceil}$.
\end{theorem}

\begin{proof}
Let $t$ be the maximum real number such that there exists a subgraph $G'$ of $G$ with $v_{G'} \geq k^2$ and $e_{G'} \geq t.v_{G'}$.
Let $G''$ be the biggest subgraph of $G$ with $e_{G''} > t.v_{G''}$. Thus, by maximality of $t$, $v_{G''} < k^2$.

Let $G_0 = G - G''$. Hence $\chi_a(G) \leq \chi_a(G_0) + k^2$.
By maximality of $G''$, for each subgraph $H$ of $G_0$, we have $e_{H} \leq t.v_{H}$.

If $t \leq \frac{v_{H} - 1}{2}$, then $e_{H} \leq (t + 1/2)(v_{H} -1)$. If $t > \frac{v_{H} - 1}{2}$, then $\frac{v_{H}}{2} < t + 1/2$. So $e_{H} \leq \frac{(v_{H} -1).v_{H}}{2} \leq (t + 1/2)(v_{H} -1)$. Therefore, $e_{H} \leq (t + 1/2)(v_{H} -1)$ for each subgraph $H$ of $G_0$.

By Nash-Williams' Theorem~\cite{nash1page}, there exists $r = \lceil t + 1/2 \rceil$ forests $F_1, F_2, \cdots, F_r$ which covers all the edges of $G_0$. We know from Theorem~\ref{arboricity.chromatic-acyclic} $\chi_a(G_0) \leq k^{s +1}$ where $s = \lceil log_pr \rceil$.

Using  inequality~(\ref{eqn wlog}) we get $log_p k \geq t - 1/2$. Therfore

$$s = \lceil log_p(\lceil t + 1/2 \rceil)\rceil \leq \lceil log_p(1 + \lceil log_pk \rceil)\rceil \leq 1 + \lceil log_plog_pk \rceil.$$

Hence $\chi_a(G) \leq k^2 +k^{2+ \lceil log_plog_pk \rceil}$.
\end{proof}

Our bound, when restricted to the case of $(m,n) = (1,0)$, slightly improves the existing bound~\cite{Kostochka97acyclicand}.

\section{On graphs with bounded maximum degree}
Recall that  $\mathcal{G}_{\Delta}$ is the family of graphs with maximum degree  $\Delta$. 
It is known that $\chi_{(1,0)}(\mathcal{G}_{\Delta}) \leq 2 \Delta^2 2^\Delta$~\cite{Kostochka97acyclicand}. 
Here we prove that $\chi_{(m,n)}(\mathcal{G}_{\Delta}) \leq 2 (\Delta-1)^p .p^{(\Delta-1)} + 2 $ for all $p = 2m+n \geq 2$ and $\Delta \geq 5$. 
Our result, restricted 
to the case $(m,n) = (1,0)$, slightly improves the upper bound of Kostochka, Sopena and Zhu~\cite{Kostochka97acyclicand}.

\begin{theorem}\label{chromatic-degree}
For the family $\mathcal{G}_{\Delta}$  of graphs with maximum degree  
$\Delta$ we have  $p^{\Delta/2} \leq \chi_{(m,n)}(\mathcal{G}_{\Delta}) \leq 2 (\Delta-1)^p. p^{(\Delta-1)} +2$ for all $p = 2m+n \geq 2$ and for all $\Delta \geq 5$.
\end{theorem}

If every subgraph of a graph $G$ have at least one vertex with degree at most $d$, then $G$ is \textit{$d$-degenerated}.  
Minimum such $d$ is the \textit{degeneracy} of $G$. To prove the above theorem we need the following result.

\begin{theorem}\label{chromatic-degree.degeneracy}
Let $\mathcal{G}'_{\Delta}$  be the family of graphs with maximum degree  $\Delta$ and degeneracy $(\Delta - 1)$. Then 
  $\chi_{(m,n)}(\mathcal{G}'_{\Delta}) \leq 2 (\Delta-1)^p .p^{(\Delta-1)}$ for all $p = 2m+n \geq 2$ and for all $\Delta \geq 5$.
\end{theorem}

To prove the above theorem we need the following lemma. 

\begin{lemma}\label{key-lemma}
There exists an  $(m,n)$-colored complete mixed graph  with property $Q^{t-1,j}_{1+(t-j)(t-2)}$ on $c = 2 (t-1)^p .p^{(t-1)}$ vertices
where $p = 2m+n \geq 2$ and $t \geq 5$. 
\end{lemma}

\begin{proof}
Let $C$ be a random $(m,n)$-colored  mixed graph with underlying complete graph. Let $u,v$ be two vertices of $C$ and the events $u \in N^{a}(v)$
for $a \in \{\pm \alpha_1, \pm \alpha_2, ...,  \pm \alpha_m,  \beta_1, \beta_2, ..., \beta_n\}$ are equiprobable and independent  with probability $\frac{1}{2m+n} = \frac{1}{p}$. 
We will show that the probability of $C$ not having property $Q^{t-1,j}_{1+(t-j)(t-2)}$ is strictly less than 1 when 
$|C| = c = 2 (t-1)^p .p^{(t-1)}$. Let $P(J,\vec{a})$ denote the probability of the event $|N^{\vec{a}}(J)| < 1+(t-j)(t-2)$ where $J$ is a $j$-tuple of $C$ and $\vec{a}$ is a $j$-vector for some $j \in \{0,1,...,t-1\}$. Call such an event a \textit{bad event}. Thus,

\begin{equation} \label{eq1}
\begin{split}
P(J, \vec{a}) & = \sum\limits_{i=0}^{(t - j)(t - 2)} {c-j \choose i} p^{-ij} (1 - p^{-j})^{c - i - j} \\
 & < (1 - p^{-j})^c \sum\limits_{i=0}^{(t - j)(t - 2)} \frac{c^i}{i!} (1 - p^{-j})^{- i - j} p^{-ij} \\
 & < 2 e^{-cp^{-j}} \sum\limits_{i=0}^{(t - j)(t - 2)} c^i \\
 & < e^{-cp^{-j}} c^{(t - j)(t - 2) + 1}.
\end{split}
\end{equation}

Let $P(B)$ denote the probability of the occurrence of at least one bad event. 
To prove this lemma it is enough to  show that $P(B) < 1$. Let $T^j$ denote the set of all $j$-tuples and $W^j$ denote the set of all $j$-vectors.  Then 

\begin{equation} \label{eq2}
\begin{split}
P(B) = \sum_{j=0}^{t - 1} \sum_{J \in T^j} \sum_{\vec{a} \in W^j} P(J, \vec{a}) & < \sum\limits_{j=0}^{t - 1} {c \choose j} p^j e^{-cp^{-j}} c^{(t - j)(t - 2) + 1} \\
 & < \sum\limits_{j=0}^{t - 1} \frac{c^j}{j!} p^j e^{-cp^{-j}} c^{(t - j)(t - 2) + 1} \\
 & = 2 \sum\limits_{j=0}^{t - 1} \frac{p^j}{2^j} \frac{2^{j-1}}{j!}  c^j e^{-cp^{-j}} c^{(t - j)(t - 2) + 1} \\
 & < 2 \sum\limits_{j=0}^{t - 1} \frac{p^j}{2^j} e^{-cp^{-j}} c^{(t - j)(t - 2) + 1 + j}.
\end{split}
\end{equation}

Consider the function $f(j) = 2 (p/2)^j e^{-cp^{-j}} c^{(t - j)(t - 2) + 1 + j}$. Observe that $f(j)$ is the $j^{th}$ 
summand  of the last sum from equation~(\ref{eq2}). 
Now
\begin{equation} \label{eq3}
\begin{split}
 \frac{f(j + 1)}{f(j)} & = \frac{p}{2} \frac{e^{(p-1)cp^{-j-1}}}{c^{t-3}} \\
 & >  \frac{p}{2} \frac{e^{(p-1)cp^{-(t-1)}}}{c^{t-3}} \\
 & >  \frac{p}{2} \left(\frac{e^{2 (p-1) (t-1)^{p-1}}}{c}\right)^{t-3} \\
\end{split}
\end{equation}

As $\frac{p - 1}{p} > \frac{1}{2}$,  $$\frac{(k-1)^{p-1}}{2} > \ln (k-1) \implies (p-1)(k-1)^{p-1} > \ln (k-1)^p.$$ 

Furthermore,  $$\frac{(p-1)}{\ln p} (k-1)^{p-1} > \frac{\ln 2}{\ln p} + (k-1) \implies (p-1) (k-1)^{p-1} > \ln(2 p^{k-1}).$$
 
 Adding the above two inequalities we get 
 $$e^{2 (p-1) (t-1)^{p-1}} > 2 (t-1)^p p^{t-1} = c.$$

Hence $\frac{f(j + 1)}{f(j)} > \frac{p}{2}$. Thus, using inequality~(\ref{eq2}) we get 
$P(B)  < \sum\limits_{j=0}^{t - 1} f(j)$.  This implies
 
\medskip 
 
\[P(B) < \begin{cases}
\frac{(p/2)^t - 1}{(p/2) - 1}f(0), & \text{  if } p > 2 \\
t f(0), & \text{  if } p = 2
\end{cases}
\]

\medskip

\textbf{Case.1:} $p > 2$.

\begin{equation} \label{eq4}
\begin{split}
 P(B) & < 2. \frac{(p/2)^t - 1}{(p/2) - 1}. \frac{c^{(t - 1)^2}}{e^{2(t - 1)^p p^{t - 1}}} \\
 & < 4. \frac{(p/2)^t - 1}{p - 2}. \left(\frac{c}{e^{2 p^{t - 1}}}\right)^{(t - 1)^p} \\
 & <  4. (p/2)^t. \left(\frac{c}{e^{2 p^{t - 1}}}\right)^{(t - 1)^p} \\
 & < \left(\frac{pc}{e^{2 p^{t - 1}}}\right)^{(t - 1)^p}
\end{split}
\end{equation}

Now, we observe that

\begin{equation*} \label{eq5}
\begin{split}
 \ln (pc) & < \ln p + \ln 2 + p\ln (t - 1) + (t - 1)\ln p \\
 & = t \ln p + p\ln (t - 1) + \ln 2 \\
 & <  tp + p(t - 1) + 2 \\
 & < 2tp < 2p^{t - 1} 
\end{split}
\end{equation*}

So from the inequality (\ref{eq4}), we can say that $P(B) < 1$ for $p >2$.

\medskip

\textbf{Case.2:} $p = 2$.

\begin{equation} \label{eq6}
\begin{split}
 P(B) & < 2t. \frac{c^{(t - 1)^2}}{e^{(t - 1)^2 2^t}} \\
 & = 2t. \left(\frac{c}{e^{2^t}}\right)^{(t - 1)^2} \\
 & < \left(\frac{2tc}{e^{2^t}}\right)^{(t - 1)^2}
\end{split}
\end{equation}

Observe that, $\ln c = 2 \ln (t-1) + t \ln 2 < 2(t - 1) + 2t = 4t - 2$.

Now, we see that 
\[
 \ln (2tc) < 4t - 2 + 2t < 6t < 2^t \\
 \implies 2tc < e^{2^t} \\
 \implies \frac{2tc}{e^{2^t}} < 1
\]

So from the inequality (\ref{eq6}), we can say that $P(B) < 1$ for $p = 2$.
\end{proof}

Now we are ready to prove Theorem~\ref{chromatic-degree.degeneracy}.

\medskip

\noindent \textit{Proof of Theorem~\ref{chromatic-degree.degeneracy}.} Suppose that $G$ is an $(m,n)$-colored mixed graph with maximum degree $\Delta$ and degeneracy 
$(\Delta-1)$.
By Lemma~\ref{key-lemma} we know that there exists an $(m,n)$-colored mixed graph $C$ with property 
$Q^{\Delta-1,j}_{1+(\Delta-j)(\Delta-2)}$ on $2 (\Delta-1)^p .p^{(\Delta-1)}$ vertices
where $p = 2m+n \geq 2$ and $\Delta \geq 5$. We will show that $G$ admits a homomorphism to $C$.

As $G$ has degeneracy $(\Delta-1)$, we can provide an ordering $v_1, v_2, ..., v_k$ of the vertices of $G$ in such a way that each vertex $v_j$ has at most $(\Delta-1)$ neighbors with lower indices.
Let $G_l$ be the  $(m,n)$-colored mixed graph induced by the vertices $v_1, v_2, ..., v_l$ from $G$ for $l \in \{1,2,...,k\}$. 
Now we will recursively construct a  homomorphism $f: G \rightarrow C$ with the following properties:

\begin{itemize}
\item[$(i)$] The partial mapping $f(v_1), f(v_2), ..., f(v_l)$ is a homomorphism of  $G_l$ to $C$ for all $l \in \{1,2,...,k\}$.

\item[$(ii)$]  For each $i > l$, all the neighbors of $v_i$ with indices less than or equal to $l$ has different images with respect to the mapping $f$.  
\end{itemize} 

Note that the base case is trivial, that is, any partial mapping $f(v_1)$ is enough. 
Suppose that the function $f$ satisfies the above properties for all $j \leq t$ where $t \in \{1,2,...,k-1\}$ is  fixed. 
Now assume that $v_{t+1}$ has $s$ neighbors with indices greater than $t+1$. 
Then $v_{t+1}$ has at most $(\Delta - s)$ neighbors with  indices less than $t+1$. 
Let $A$ be the set of neighbors of $v_{t+1}$   with  indices greater than $t+1$.
 Let $B$ be the set of vertices   with indices at most $t$ and with at least one neighbor in $A$. 
 Note that as each vertex of $A$ is a neighbor 
 of $v_{t+1}$ and has at most $\Delta-1$ neighbors with lesser indices, $|B| = (\Delta-2)|A| = s(\Delta-2)$. 
 Let $D$ be the set of possible  options for $f(v_{t+1})$ such that the partial mapping 
 is a homomorphism of $G_{t+1}$ to $C$. 
 As $C$ has property $Q^{\Delta-1,j}_{1+(\Delta-j)(\Delta-2)}$ we have $|C| \geq 1+s(\Delta-1)$. 
 So the set $D \setminus B$ is non-empty. 
 Thus, choose any vertex from $D \setminus B$ as the image $f(v_{t+1})$. 
 Note that this partial mapping satisfies the required conditions.  \hfill $ \square $

\medskip

Finally, we are ready to prove Theorem~\ref{chromatic-degree}.

\medskip

\noindent \textit{Proof of Theorem~\ref{chromatic-degree}.}  First we will prove the lower bound. 
Let $G_t$ be a $\Delta$ regular graph on $t$ vertices. Thus, $G_t$ has $\frac{t \Delta}{2}$  edges. Then 
we have $$k_t = \chi_{(m,n)}(G_t)  \geq \frac{p^{\Delta /2}}{p^{\left. {k_t \choose 2} \middle/  t \right.}}$$ 
using inequality~(\ref{eqn wolog}) (see Section~\ref{sec acyclic}). If $\chi_{(m,n)}(G_t) \geq p^{\Delta /2} $ for some $t$, then we are done. 
Otherwise, $\chi_{(m,n)}(G_t) = k_t$ is bounded. In that case, if $t$ is sufficiently large, then $\chi_{(m,n)}(G_t) \geq p^{\Delta /2}$ as 
$\chi_{(m,n)}(G_t)$ is a positive integer.

\medskip

Let $G = (V, A \cup E)$ be a connected $(m,n)$-colored mixed graph with maximum degree $\Delta \geq 5$ and $p = 2m+n \geq 2$. 
If $G$ has a vertex of degree at most $(\Delta-1)$ then it has degeneracy at most $(\Delta-1)$. In that case  by Theorem~\ref{chromatic-degree}
we are done. 

Otherwise, $G$ is $\Delta$ regular. In that case, remove an edge $uv$ of $G$ to obtain the graph $G'$. Note that $G'$ has 
maximum degree at most $\Delta$  and has degeneracy at most $(\Delta-1)$. Therefore, by Theorem~\ref{chromatic-degree} there exists an 
$(m,n)$-colored complete mixed graph $C$ on 
$2 (\Delta-1)^p .p^{(\Delta-1)}$ vertices to which $G'$ admits a $f$ homomorphism to. 
Let $G''$ be the graph obtained by deleting the vertices $u$ and $v$ of $G'$. Note that the homomorphism $f$ restricted to $G''$ is 
 a homomorphism $f_{res}$ of $G''$ to $C$. Now include two new vertices $u'$ and $v'$ to $C$ and obtain a new graph $C'$. 
Color the edges or arcs between the vertices of $C$ and $\{u',v'\}$ in such a way so that we can extend the homomorphism $f_{res}$ to a homomorphism
$f_{ext}$ 
of $G$ to $C'$  where $f_{ext}(u) = u'$, $f_{ext}(v) = v'$ and $f_{ext}(x) = f_{res}(x)$ for all $x \in V(G) \setminus \{u,v\}$. 
It is easy to note that the above mentioned process is possible. 

Thus, every connected $(m,n)$-colored mixed graph with maximum degree $\Delta$ admits a homomorphism to $C'$. 
\hfill $ \square$

\bibliographystyle{abbrv}
\bibliography{NSS14}

\begin{thebibliography}{1}

\bibitem{Marshall-edgecoloring}
N.~Alon and T.~H. Marshall.
\newblock Homomorphisms of edge-colored graphs and {C}oxeter groups.
\newblock {\em Journal of Algebraic Combinatorics}, 8(1):5--13, 1998.

\bibitem{Borodinacyclic}
O.~V. Borodin.
\newblock On acyclic colorings of planar graphs.
\newblock {\em Discrete Mathematics}, 25(3):211--236, 1979.

\bibitem{Kostochka97acyclicand}
A.~V. Kostochka, {\'E}.~Sopena, and X.~Zhu.
\newblock Acyclic and oriented chromatic numbers of graphs.
\newblock {\em Journal of Graph Theory}, 24:331--340, 1997.

\bibitem{nash1page}
C.~S.~J. Nash-Williams.
\newblock Decomposition of finite graphs into forests.
\newblock {\em Journal of the London Mathematical Society}, 1(1):12--12, 1964.

\bibitem{raspaud_and_nesetril}
J.~Ne{\v{s}}et{\v{r}}il and A.~Raspaud.
\newblock Colored homomorphisms of colored mixed graphs.
\newblock {\em Journal of Combinatorial Theory, Series B}, 80(1):147--155,
  2000.

\bibitem{Ochem_negativeresults}
P.~Ochem.
\newblock Negative results on acyclic improper colorings.
\newblock In {\em European Conference on Combinatorics (EuroComb 2005)}, pages
  357--362, 2005.

\bibitem{planar80}
A.~Raspaud and {\'E}.~Sopena.
\newblock Good and semi-strong colorings of oriented planar graphs.
\newblock {\em Information Processing Letters}, 51(4):171--174, 1994.

\bibitem{sopena_updated_survey}
E.~Sopena.
\newblock Homomorphisms and colourings of oriented graphs: An updated survey.
\newblock {\em Discrete Mathematics}, 2015 (in press).

\end{thebibliography}

\end{document}